\documentclass{birkjour}
\usepackage[utf8]{inputenc}
\usepackage[english]{babel}
\usepackage[T1]{fontenc}
\usepackage{amsmath,amssymb,amsthm} 

\newcommand{\im}{\mathrm{i}}
\newcommand{\N}{\mathbb{N}}

\newcommand{\R}{\mathbb{R}}
\newcommand{\C}{\mathbb{C}}
\newcommand{\defeq}{:=}
\newcommand{\tens}{\otimes}
\newcommand{\ctens}{\hat{\otimes}}
\newcommand{\cH}{\mathcal{H}}
\newcommand{\rL}{\mathrm{L}}
\newcommand{\xd}{\mathrm{d}}

\newcommand{\cA}{\mathcal{A}}
\newcommand{\cS}{\mathcal{S}}
\newcommand{\one}{\mathbf{1}}
\newcommand{\ds}{\circ}
\newcommand{\obs}{\mathcal{O}}
\newcommand{\aglue}{\diamond}
\newcommand{\tord}{\mathbf{T}}
\newcommand{\no}[1]{:#1:}
\newcommand{\ano}[1]{\blacktriangleleft #1 \blacktriangleright}

\theoremstyle{definition}
\newtheorem{dfn}{Definition}[section]

\theoremstyle{plain}

\newtheorem{prop}[dfn]{Proposition}

\begin{document}

\title[Observables in the General Boundary Formulation]
{Observables in the\\ General Boundary Formulation}
\author[Robert Oeckl]{Robert Oeckl}
\address{%
Instituto de Matemáticas\\
Universidad Nacional Autónoma de México\\
Campus Morelia\\
C.P. 58190, Morelia, Michoacán\\
Mexico}
\email{robert@matmor.unam.mx}


\begin{abstract}
We develop a notion of quantum observable for the general boundary formulation of quantum theory. This notion is adapted to spacetime regions rather than to hypersurfaces and naturally fits into the topological quantum field theory like axiomatic structure of the general boundary formulation. We also provide a proposal for a generalized concept of expectation value adapted to this type of observable. We show how the standard notion of quantum observable arises as a special case together with the usual expectation values. We proceed to introduce various quantization schemes to obtain such quantum observables including path integral quantization (yielding the time-ordered product), Berezin-Toeplitz (antinormal ordered) quantization and normal ordered quantization, and discuss some of their properties.

\end{abstract}

\maketitle

\section{Motivation: Commutation relations and quantum field theory}

In standard quantum theory one is used to think of observables as encoded in operators on the Hilbert space of states. The algebra formed by these is then seen as encoding fundamental structure of the quantum theory. Moreover, this algebra often constitutes the primary target of quantization schemes that aim to produce a quantum theory from a classical theory. Commutation relations in this algebra then provide a key ingredient of correspondence principles between a classical theory and its quantization.
We shall argue in the following that while this point of view is natural in a non-relativistic setting, it is less compelling in a special relativistic setting and becomes questionable in a general relativistic setting.\footnote{By ``general relativistic setting'' we shall understand here a context where the metric of spacetime is a dynamical object, but which is not necessarily limited to Einstein's theory of General Relativity.}

In non-relativistic quantum mechanics (certain) operators
correspond to measurements that can be applied at any given time, meaning that the
measurement is performed at that time. Let us say we consider the measurement of two
quantities, one associated with the operator $A$ and another associated with the
operator $B$. In particular, we can then also say which operator is associated with
the \emph{consecutive} measurement of both quantities. If we first measure $A$ and then
$B$ the operator is the product $B A$, and if we first measure $B$ and then $A$ the
operator is the product $A B$.\footnote{The notion of composition of measurements considered here is one where the output value generated by the composite measurement can be understood as the product of output values of the individual measurements, rather than one where one would obtain separate output values for both measurements.} Hence, the operator product has the operational meaning
of describing the \emph{temporal composition} of measurements. One of the key features of quantum
mechanics is of course the fact that in general $A B\neq B A$, i.e., a different
temporal ordering of measurements leads to a different outcome.

The treatment of operators representing observables is different in quantum field theory. Here, such operators are \emph{labeled} with the time at which they are applied. For example, we write $\phi(t,x)$ for a field operator at time $t$. Hence, if we want to combine the measurement processes
associated with operators $\phi(t,x)$ and $\phi(t',x')$ say, there is only one
operationally meaningful way to do so. The operator associated with the combined process
is the \emph{time-ordered} product of the two operators, $\tord{\phi(t,x) \phi(t',x')}$.
Of course, this time-ordered product is commutative since the information about the
temporal ordering of the processes associated with the operators is already contained
in their labels. Nevertheless, in traditional treatments of quantum field theory
one first constructs a non-commutative algebra of field operators starting with equal-time commutation relations. Since the concept of equal-time hypersurface is not Poincaré invariant, one then
goes on to generalize these commutation relations to field operators at different times. In particular, one finds that for two localized operators $A(t,x)$ and $B(t',x')$, the commutator obeys
\begin{equation}
 [A(t,x),B(t',x')]= 0\qquad\text{if $(t,x)$ and $(t',x')$ are spacelike separated} ,
\label{eq:causcom}
\end{equation}
which is indeed a Poincaré invariant condition.
The time-ordered product is usually treated as a concept that is \emph{derived} from
the non-commutative operator product. From this point of view, condition (\ref{eq:causcom}) serves to make sure that it is well defined and does not depend on the inertial frame. Nevertheless, it is the former and not the latter that has a direct operational meaning. Indeed, essentially all the predictive power of quantum field theory derives from the amplitudes and the S-matrix which are defined entirely in terms of time-ordered products. On the other hand, the non-commutative operator product can be recovered from the time-ordered product. Equal-time commutation relations can be obtained as the limit,
\begin{equation*}
 [A(t,x),B(t,x')]=\lim_{\epsilon\to +0}
  \tord{A(t+\epsilon,x) B(t-\epsilon,x')} - \tord{A(t-\epsilon,x) B(t+\epsilon,x')} .
\end{equation*}
The property (\ref{eq:causcom}) can then be seen as arising from the transformation properties of this limit and its non-equal time generalization.

We conclude that in a special relativistic setting, there are good reasons to regard the time-ordered product of observables as more fundamental than the non-commutative operator product. This suggests to try to formulate the theory of observables in terms of the former rather than the latter. In a (quantum) general relativistic setting with no predefined background metric a condition such as (\ref{eq:causcom}) makes no longer sense, making the postulation of a non-commutative algebra structure for observables even more questionable.

In this paper we shall consider a proposal for a concept of quantum observable that takes these concerns into account. The wider framework in which we embed this is the general boundary formulation of quantum theory (GBF) \cite{Oe:GBQFT}. We start in Section~\ref{sec:rev} with a short review of the relevant ingredients of the GBF. In Section~\ref{sec:obs} we introduce a concept of quantum observable in an axiomatic way, provide a suitably general notion of expectation value and show how standard concepts of quantum observable and expectation values arise as special cases. In Section~\ref{sec:quant} we consider different quantization prescriptions of classical observables that produce such quantum observables, mainly in a field theoretic context.

\section{Short review of the general boundary formulation}
\label{sec:rev}

\subsection{Core axioms}
\label{sec:caxioms}

The basic data of a general boundary quantum field theory consists of two types: geometric objects that encode a basic structure of spacetime and algebraic objects that encode notions of quantum states and amplitudes. The algebraic objects are assigned to the geometric objects in such a way that the core axioms of the general boundary formulation are satisfied. These may be viewed as a special variant of the axioms of a topological quantum field theory \cite{Ati:tqft}. They have been elaborated, with increasing level of precision, in \cite{Oe:boundary, Oe:GBQFT,Oe:2dqym,Oe:holomorphic}. In order for this article to be reasonably self-contained, we repeat below the version from \cite{Oe:holomorphic}.

The geometric objects are of two kinds:
\begin{description}
\item[Regions] These are (certain) oriented manifolds of dimension $d$ (the spacetime dimension), usually with boundary.
\item[Hypersurfaces] These are (certain) oriented manifolds of dimension $d-1$, here assumed without boundary.\footnote{The setting may be generalized to allow for hypersurfaces with boundaries along the lines of \cite{Oe:2dqym}. However, as the required modifications are of little relevance in the context of the present paper, we restrict to the simpler setting.}
\end{description}
Depending on the theory to be modeled, the manifolds may carry additional structure such as that of a Lorentzian metric in the case of quantum field theory. For more details see the references mentioned above. The core axioms may be stated as follows:

\begin{itemize}
\item[(T1)] Associated to each hypersurface $\Sigma$ is a complex
  separable Hilbert space $\cH_\Sigma$, called the \emph{state space} of
  $\Sigma$. We denote its inner product by
  $\langle\cdot,\cdot\rangle_\Sigma$.
\item[(T1b)] Associated to each hypersurface $\Sigma$ is a conjugate linear
  isometry $\iota_\Sigma:\cH_\Sigma\to\cH_{\bar{\Sigma}}$. This map
  is an involution in the sense that $\iota_{\bar{\Sigma}}\circ\iota_\Sigma$
  is the identity on  $\cH_\Sigma$.
\item[(T2)] Suppose the hypersurface $\Sigma$ decomposes into a disjoint
  union of hypersurfaces $\Sigma=\Sigma_1\cup\cdots\cup\Sigma_n$. Then,
  there is an isometric isomorphism of Hilbert spaces
  $\tau_{\Sigma_1,\dots,\Sigma_n;\Sigma}:\cH_{\Sigma_1}\ctens\cdots\ctens\cH_{\Sigma_n}\to\cH_\Sigma$.
  The composition of the maps $\tau$ associated with two consecutive
  decompositions is identical to the map $\tau$ associated to the
  resulting decomposition.
\item[(T2b)] The involution $\iota$ is compatible with the above
  decomposition. That is,
  $\tau_{\bar{\Sigma}_1,\dots,\bar{\Sigma}_n;\bar{\Sigma}}
  \circ(\iota_{\Sigma_1}\ctens\cdots\ctens\iota_{\Sigma_n}) 
  =\iota_\Sigma\circ\tau_{\Sigma_1,\dots,\Sigma_n;\Sigma}$.
\item[(T4)] Associated with each region $M$ is a linear map
  from a dense subspace $\cH_{\partial M}^\ds$ of the state space
  $\cH_{\partial M}$ of its boundary $\partial M$ (which carries the
  induced orientation) to the complex
  numbers, $\rho_M:\cH_{\partial M}^\ds\to\C$. This is called the
  \emph{amplitude map}.
\item[(T3x)] Let $\Sigma$ be a hypersurface. The boundary $\partial\hat{\Sigma}$ of the associated empty region $\hat{\Sigma}$ decomposes into the disjoint union $\partial\hat{\Sigma}=\bar{\Sigma}\cup\Sigma'$, where $\Sigma'$ denotes a second copy of $\Sigma$. Then, $\tau_{\bar{\Sigma},\Sigma';\partial\hat{\Sigma}}(\cH_{\bar{\Sigma}}\tens\cH_{\Sigma'})\subseteq\cH_{\partial\hat{\Sigma}}^\ds$. Moreover, $\rho_{\hat{\Sigma}}\circ\tau_{\bar{\Sigma},\Sigma';\partial\hat{\Sigma}}$ restricts to a bilinear pairing $(\cdot,\cdot)_\Sigma:\cH_{\bar{\Sigma}}\times\cH_{\Sigma'}\to\C$ such that $\langle\cdot,\cdot\rangle_\Sigma=(\iota_\Sigma(\cdot),\cdot)_\Sigma$.
\item[(T5a)] Let $M_1$ and $M_2$ be regions and $M\defeq M_1\cup M_2$ be their disjoint union. Then $\partial M=\partial M_1\cup \partial M_2$ is also a disjoint union and $\tau_{\partial M_1,\partial M_2;\partial M}(\cH_{\partial M_1}^\ds\tens \cH_{\partial M_2}^\ds)\subseteq \cH_{\partial M}^\ds$. Then, for all $\psi_1\in\cH_{\partial M_1}^\ds$ and $\psi_2\in\cH_{\partial M_2}^\ds$,
\begin{equation}
 \rho_{M}\circ\tau_{\partial M_1,\partial M_2;\partial M}(\psi_1\tens\psi_2)= \rho_{M_1}(\psi_1)\rho_{M_2}(\psi_2) .
\end{equation}
\item[(T5b)] Let $M$ be a region with its boundary decomposing as a disjoint union $\partial M=\Sigma_1\cup\Sigma\cup \overline{\Sigma'}$, where $\Sigma'$ is a copy of $\Sigma$. Let $M_1$ denote the gluing of $M$ with itself along $\Sigma,\overline{\Sigma'}$ and suppose that $M_1$ is a region. Note $\partial M_1=\Sigma_1$. Then, $\tau_{\Sigma_1,\Sigma,\overline{\Sigma'};\partial M}(\psi\tens\xi\tens\iota_\Sigma(\xi))\in\cH_{\partial M}^\ds$ for all $\psi\in\cH_{\partial M_1}^\ds$ and $\xi\in\cH_\Sigma$. Moreover, for any orthonormal basis $\{\xi_i\}_{i\in I}$ of $\cH_\Sigma$, we have for all $\psi\in\cH_{\partial M_1}^\ds$,
\begin{equation}
 \rho_{M_1}(\psi)\cdot c(M;\Sigma,\overline{\Sigma'})
 =\sum_{i\in I}\rho_M\circ\tau_{\Sigma_1,\Sigma,\overline{\Sigma'};\partial M}(\psi\tens\xi_i\tens\iota_\Sigma(\xi_i)),
\label{eq:glueax1}
\end{equation}
where $c(M;\Sigma,\overline{\Sigma'})\in\C\setminus\{0\}$ is called the \emph{gluing anomaly factor} and depends only on the geometric data.
\end{itemize}

As in \cite{Oe:holomorphic} we omit in the following the explicit mention of the maps $\tau$.

\subsection{Amplitudes and probabilities}
\label{sec:ampprob}

In standard quantum theory transition amplitudes can be used to encode measurements. The setup, in its simplest form, involves an initial state $\psi$ and a final state $\eta$. The initial state encodes a \emph{preparation} of or \emph{knowledge} about the measurement, while the final state encodes a \emph{question} about or \emph{observation} of the system. The modulus square $|\langle\eta,U\psi\rangle|^2$, where $U$ is the time-evolution operator between initial and final time, is then the probability for the answer to the question to be affirmative. (We assume states to be normalized.) This is a \emph{conditional probability} $P(\eta|\psi)$, namely the probability to observe $\eta$ given that $\psi$ was prepared.

In the GBF this type of measurement setup generalizes considerably.\footnote{Even in standard quantum theory, generalizations are possible which involve subspaces of the Hilbert space instead of states. A broader analysis of this situation shows that formula (\ref{eq:prob}) is a much milder generalization of standard probability rules than might seem at first sight, see \cite{Oe:GBQFT}.} Given a spacetime region $M$, a \emph{preparation} of or \emph{knowledge} about the measurement is encoded through a closed subspace $\cS$ of the boundary Hilbert space $\cH_{\partial M}$. Similarly, the \emph{question} or \emph{observation} is encoded in another closed subspace $\cA$ of $\cH_{\partial M}$. The \emph{conditional probability} for observing $\cA$ given that $\cS$ is prepared (or known to be the case) is given by the following formula \cite{Oe:GBQFT,Oe:probgbf}:
\begin{equation}
   P(\cA|\cS)=\frac{\|\rho_M\circ P_\cS\circ P_\cA\|^2}
   {\|\rho_M\circ P_\cS\|^2} .
\label{eq:prob}
\end{equation}
Here $P_{\cS}$ and $P_{\cA}$ are the orthogonal projectors onto the subspaces $\cS$ and $\cA$ respectively. $\rho_M\circ P_{\cS}$ is the linear map $\cH_{\partial M}\to\C$ given by the composition of the amplitude map $\rho_M$ with the projector $P_{\cS}$. A requirement for (\ref{eq:prob}) to make sense is that this composed map is continuous, but does not vanish. (The amplitude map $\rho_M$ is generically not continuous.) That is, $\cS$ must be neither ``too large'' nor ``too small''. Physically this means that $\cS$ must on the one hand be sufficiently restrictive while on the other hand not imposing an impossibility. The continuity of $\rho_M\circ P_{\cS}$ means that it is an element in the dual Hilbert space $\cH_{\partial M}^*$. The norm in $\cH_{\partial M}^*$ is denoted in formula (\ref{eq:prob}) by $\|\cdot\|$. With an analogous explanation for the numerator the mathematical meaning of (\ref{eq:prob}) is thus clear.

In \cite{Oe:GBQFT}, where this probability interpretation of the GBF was originally proposed, the additional assumption $\cA\subseteq \cS$ was made, and with good reason. Physically speaking, this condition enforces that we only ask questions in a way that takes into account fully what we already know. Since it is of relevance in the following, we remark that formula (\ref{eq:prob}) might be rewritten in this case as follows:
\begin{equation}
   P(\cA|\cS)=\frac{\langle \rho_M\circ P_\cS,\rho_M\circ P_\cA\rangle}
   {\|\rho_M\circ P_\cS\|^2} .
\label{eq:prob2}
\end{equation}
Here the inner product $\langle\cdot,\cdot\rangle$ is the inner product of the dual Hilbert space $\cH_{\partial M}^*$. Indeed, whenever $P_{\cS}$ and $P_{\cA}$ commute, (\ref{eq:prob2}) coincides with (\ref{eq:prob}).

\subsection{Recovery of standard transition amplitudes and probabilities}
\label{sec:stdampl}

We briefly recall in the following how standard transition amplitudes are recovered from amplitude functions. Similarly, we recall how standard transition probabilities arise as special cases of the formula (\ref{eq:prob}). Say $t_1$ is some initial time and $t_2> t_1$ some final time, and we consider the spacetime region $M=[t_1,t_2]\times\R^3$ in Minkowski space. $\partial M$ is the disjoint union $\Sigma_1\cup\bar{\Sigma}_2$ of hypersurfaces of constant $t=t_1$ and $t=t_2$ respectively. We have chosen the orientation of $\Sigma_2$ here to be opposite to that induced by $M$, but equal (under time-translation) to that of $\Sigma_1$. Due to axioms (T2) and (T1b), we can identify the Hilbert space $\cH_{\partial M}$ with the tensor product $\cH_{\Sigma_1}\ctens\cH_{\Sigma_2}^*$. The amplitude map $\rho_M$ associated with $M$ can thus be viewed as a linear map $\cH_{\Sigma_1}\ctens\cH_{\Sigma_2}^*\to\C$.

In the standard formalism, we have on the other hand a single Hilbert space $H$ of states and a unitary time-evolution map $U(t_1,t_2):H\to H$. To relate the two settings we should think of $H$, $\cH_{\Sigma_1}$ and $\cH_{\Sigma_2}$ as really identical (due to time-translation being an isometry). Then, for any $\psi,\eta\in H$, the amplitude map $\rho_M$ and the operator $U$ are related as
\begin{equation}
 \rho_M(\psi\tens\iota(\eta))=\langle\eta, U(t_1,t_2)\psi\rangle .
\label{eq:amtampl}
\end{equation}

Consider now a measurement in the same spacetime region, where an initial state $\psi$ is prepared at time $t_1$ and a final state $\eta$ is tested at time $t_2$. The standard formalism tells us that the probability for this is (assuming normalized states):
\begin{equation}
P(\eta|\psi)=|\langle \eta, U(t_1,t_2)\psi\rangle|^2.
\label{eq:ptampl}
\end{equation}
In the GBF, the preparation of $\psi$ and observation of $\eta$ are encoded in the following subspaces of $\cH_{\Sigma_1}\ctens\cH_{\Sigma_2}^*$:
\begin{equation}
 \cS=\{\psi\tens\xi: \xi\in\cH_{\Sigma_2}\}\quad\text{and}\quad
 \cA=\{\lambda\psi\tens\iota(\eta): \lambda\in\C\} .
\end{equation}
Using (\ref{eq:amtampl}) one can easily show that then $P(\cA|\cS)=P(\eta|\psi)$, i.e., the expressions (\ref{eq:prob}) and (\ref{eq:ptampl}) coincide. This remains true if, alternatively, we define $\cA$ disregarding the knowledge encoded in $\cS$, i.e., as
\begin{equation}
 \cA=\{\xi\tens\iota(\eta): \xi\in\cH_{\Sigma_1}\} .
\end{equation}

\section{A conceptual framework for observables}
\label{sec:obs}

\subsection{Axiomatics}
\label{sec:obsax}

Taking account of the fact that realistic measurements are extended both in space as well as in time, it appears sensible to locate also the mathematical objects that represent observables in spacetime regions. This is familiar for example from algebraic quantum field theory, while being in contrast to idealizing measurements as happening at instants of time as in the standard formulation of quantum theory.

Mathematically, we model an observable associated with a given spacetime region $M$ as a replacement of the corresponding amplitude map $\rho_M$. That is, an observable in $M$ is a linear map $\cH_{\partial M}^\ds\to\C$, where $\cH_{\partial M}^\ds$ is the dense subspace of $\cH_{\partial M}$ appearing in core axiom (T4). Not any such map needs to be an observable though. Which map exactly qualifies as an observable may generally depend on the theory under consideration.

\begin{itemize}
\item[(O1)] Associated to each spacetime region $M$ is a real vector space $\obs_M$ of linear maps $\cH_{\partial M}^\ds\to\C$, called \emph{observable maps}. In particular, $\rho_M\in\obs_M$.
\end{itemize}

The most important operation that can be performed with observables is that of \emph{composition}. This composition is performed exactly in the same way as prescribed for amplitude maps in core axioms (T5a) and (T5b). This leads to an additional condition on the spaces $\obs_M$ of observables, namely that they be closed under composition.

\begin{itemize}
\item[(O2a)] Let $M_1$ and $M_2$ be regions as in (T5a) and $O_1\in\obs_{M_1}$ and $O_2\in\obs_{M_2}$. Then, there is $O_3\in\obs_{M_1\cup M_2}$ such that for all $\psi_1\in\cH_{\partial M_1}^\ds$ and $\psi_2\in\cH_{\partial M_2}^\ds$,
\begin{equation}
 O_3(\psi_1\tens\psi_2)= \rho_{M_1}(\psi_1)\rho_{M_2}(\psi_2) .
\end{equation}
\item[(O2b)] Let $M$ be a region with its boundary decomposing as a disjoint union $\partial M=\Sigma_1\cup\Sigma\cup \overline{\Sigma'}$ and $M_1$ given as in (T5b) and $O\in\obs_{M}$. Then, there exists $O_1\in\obs_{M_1}$ such that for any orthonormal basis $\{\xi_i\}_{i\in I}$ of $\cH_\Sigma$ and for all $\psi\in\cH_{\partial M_1}^\ds$,
\begin{equation}
 O_1(\psi)\cdot c(M;\Sigma,\overline{\Sigma'})
 =\sum_{i\in I}O(\psi\tens\xi_i\tens\iota_\Sigma(\xi_i)) .
\end{equation}
\end{itemize}
We generally refer to the gluing operations of observables of the types described in (O2a) and (O2b) as well as their iterations and combinations as \emph{compositions of observables}. Physically, the composition is meant to represent the combination of measurements. Combination is here to be understood as in classical physics, when the product of observables is taken.

\subsection{Expectation values}

As in the standard formulation of quantum theory, the \emph{expectation value} of an observable depends on a preparation of or knowledge about a system. As recalled in Section~\ref{sec:ampprob}, this is encoded for a region $M$ in a closed subspace $\cS$ of the boundary Hilbert space $\cH_{\partial M}$. Given an observable $O\in\obs_M$ and a closed subspace $\cS\subseteq\cH_{\partial M}$, the expectation value of $O$ with respect to $\cS$ is defined as
\begin{equation}
\langle O\rangle_{\cS}\defeq\frac{\langle\rho_M\circ P_{\cS}, O\rangle}{\|\rho_M\circ P_{\cS}\|^2} .
\label{eq:expval}
\end{equation}
We use notation here from Section~\ref{sec:ampprob}. Also, as there we need $\rho_M\circ P_{\cS}$ to be continuous and different from zero for the expectation value to make sense.

We proceed to make some remarks about the motivation for postulating the expression (\ref{eq:expval}). Clearly, the expectation value must be linear in the observable. Another important requirement is that we would like probabilities in the sense of Section~\ref{sec:ampprob} to arise as a special case of expectation values. Indeed, given a closed subspace $\cA$ of $\cH_{\partial M}$ and setting $O=\rho_M\circ P_{\cA}$ we see that expression (\ref{eq:expval}) reproduces exactly expression (\ref{eq:prob2}). At least in the case where the condition $\cA\subseteq\cS$ is met, this coincides with expression (\ref{eq:prob}) and represents the conditional probability to observe $\cA$ given $\cS$.

\subsection{Recovery of standard observables and expectation values}
\label{sec:recobsexp}

Of course, it is essential that the present proposal for implementing observables in the GBF can reproduce observables and their expectation values as occurring in the standard formulation of quantum theory. There observables are associated to instants of time, i.e., equal-time hypersurfaces. To model these we use ``infinitesimally thin'' regions, also called \emph{empty regions}, which geometrically speaking are really hypersurfaces, but are treated as regions.

Concretely, consider the equal-time hypersurface at time $t$ in Minkowski space, i.e., $\Sigma=\{t\}\times\R^3$. We denote the empty region defined by the hypersurface $\Sigma$ as $\hat{\Sigma}$. The relation between an observable map $O\in\obs_M$ and the corresponding operator $\tilde{O}$ is then analogous to the relation between the amplitude map and the time-evolution operator as expressed in equation (\ref{eq:amtampl}). By definition, $\partial \hat{\Sigma}$ is equal to the disjoint union $\Sigma\cup\bar{\Sigma}$ so that $\cH_{\partial \hat{\Sigma}}=\cH_{\Sigma}\ctens\cH_{\Sigma}^*$. The Hilbert space $\cH_{\Sigma}$ is identified with the conventional Hilbert space $H$ and for $\psi,\eta\in H$ we require
\begin{equation}
 O(\psi\tens\iota(\eta))=\langle\eta,\tilde{O}\psi\rangle_\Sigma .
\label{eq:obscor}
\end{equation}
Note that we can glue two copies of $\hat{\Sigma}$ together, yielding again a copy of $\hat{\Sigma}$. The induced composition of observable maps then translates via (\ref{eq:obscor}) precisely to the composition of the corresponding operators. In this way we recover the usual operator product for observables of the standard formulation. 

Consider now a normalized state $\psi\in H=\cH_{\Sigma}$ encoding a preparation. This translates in the GBF language to the subspace $\cS=\{\psi\tens\xi:\xi\in\cH_{\Sigma}^*\}$ of $\cH_{\partial \hat{\Sigma}}$ as reviewed in Section~\ref{sec:stdampl}. The amplitude map $\rho_{\hat{\Sigma}}$ can be identified with the inner product of $H=\cH_{\Sigma}$ due to core axiom (T3x). Thus, $\rho_{\hat{\Sigma}}\circ P_{\cS}(\xi\tens\iota(\eta))=\langle\eta,P_{\psi}\xi\rangle_\Sigma$, where $P_{\psi}$ is the orthogonal projector in $\cH_{\Sigma}$ onto the subspace spanned by $\psi$. This makes it straightforward to evaluate the denominator of (\ref{eq:expval}). Let $\{\xi_i\}_{i\in\N}$ be an orthonormal basis of $\cH_{\Sigma}$, which moreover we choose for convenience such that $\xi_1=\psi$. Then,
\begin{multline}
\|\rho_M\circ P_{\cS}\|^2
=\sum_{i,j=1}^\infty \left|\rho_M\circ P_{\cS} (\xi_i\tens\iota(\xi_j))\right|^2
=\sum_{i,j=1}^\infty \left|\langle \xi_j,P_{\psi} \xi_i\rangle_\Sigma\right|^2
=1 .
\end{multline}
For the numerator of (\ref{eq:expval}) we observe
\begin{multline}
\langle\rho_M\circ P_{\cS}, O\rangle
=\sum_{i,j=1}^\infty \overline{\rho_M\circ P_{\cS} (\xi_i\tens\iota(\xi_j))}\,
 O(\xi_i\tens\iota(\xi_j))\\
=\sum_{i,j=1}^\infty \langle P_{\psi} \xi_i, \xi_j\rangle_\Sigma\,
 \langle\xi_j,\tilde{O}\xi_i\rangle_\Sigma
=\langle \psi, \tilde{O}\psi\rangle_\Sigma .
\end{multline}
Hence, the GBF formula (\ref{eq:expval}) recovers in this case the conventional expectation value of $\tilde{O}$ with respect to the state $\psi$.

\section{Quantization}
\label{sec:quant}

We turn in this section to the problem of the \emph{quantization} of classical observables. On the one hand, we consider the question of how specific quantization schemes that produce Hilbert spaces and amplitude functions satisfying the core axioms can be extended to produce observables. On the other hand, we discuss general features of quantization schemes for observables and the relation to conventional schemes.

\subsection{Schrödinger-Feynman quantization}

Combining the Schrödinger representation with the Feynman path integral yields a quantization scheme that produces Hilbert spaces for hypersurfaces and amplitude maps for regions in a way that ``obviously'' satisfies the core axioms \cite{Oe:boundary,Oe:timelike,Oe:KGtl}. We shall see that it is quite straightforward to include observables into this scheme. Moreover, the resulting quantization can be seen to be in complete agreement with the results of standard approaches to quantum field theory.

We recall that in this scheme states on a hypersurface $\Sigma$ arise as wave functions on the space space $K_{\Sigma}$ of field configurations on $\Sigma$. These form a Hilbert space $\cH_{\Sigma}$ of square-integrable functions with respect to a (fictitious) translation invariant measure $\mu_{\Sigma}$:
\begin{equation}
\langle \psi',\psi\rangle_{\Sigma} \defeq\int_{K_\Sigma} \overline{\psi'(\varphi)}\psi(\varphi)\,\xd\mu_{\Sigma}(\varphi) .
\label{eq:sip}
\end{equation}
The amplitude map for a region $M$ arises as the Feynman path integral,
\begin{equation}
\rho_M(\psi)\defeq\int_{K_{M}} \psi\left(\phi|_\Sigma\right) e^{\im S_M(\phi)}\,\xd\mu_{M}(\phi) ,
\label{eq:api}
\end{equation}
where $S_M$ is the action evaluated in $M$, and $K_M$ is the space of field configurations in $M$.

The Feynman path integral is of course famous for resisting a rigorous definition and it is a highly non-trivial task to make sense of expressions (\ref{eq:api}) or even (\ref{eq:sip}) in general. Nevertheless, much of text-book quantum field theory relies on the Feynman path integral and can be carried over to the present context relatively easily for equal-time hypersurfaces in Minkowski space and regions bounded by such. Moreover, for other special types of regions and hypersurfaces this quantization program has also been successfully carried through for linear or perturbative quantum field theories. Notably, this includes timelike hypersurfaces \cite{Oe:timelike,Oe:KGtl} and has led to a widening of the concept of an asymptotic S-matrix \cite{CoOe:spsmatrix,CoOe:smatrixgbf}.

We proceed to incorporate observables into the quantization scheme. To this end, a classical observable $F$ in a region $M$ is modeled as a real (or complex) valued function on $K_M$. According to Section~\ref{sec:obsax} the quantization of $F$, which we denote here by $\rho_M^F$, must be a linear map $\cH_{\partial M}^\ds\to\C$. We define it as
\begin{equation}
 \rho_M^F(\psi)\defeq\int_{K_{M}} \psi\left(\phi|_\Sigma\right) F(\phi)\, e^{\im S_M(\phi)}\,\xd\mu_{M}(\phi) .
\label{eq:opi}
\end{equation}
Before we proceed to interpret this formula in terms of text-book quantum field theory language, we emphasize a key property of this quantization prescription. Suppose we have disjoint, but adjacent spacetime regions $M_1$ and $M_2$ supporting classical observables $F_1:K_{M_1}\to\R$ and $F_2:K_{M_2}\to\R$ respectively. Applying first the operation of (O2a) and then that of (O2b), we can compose the corresponding quantum observables $\rho_{M_1}^{F_1}$ and $\rho_{M_2}^{F_2}$ to a new observable, which we shall denote $\rho_{M_1}^{F_1}\aglue \rho_{M_2}^{F_2}$, supported on the spacetime region $M\defeq M_1\cup M_2$. On the other hand, the classical observables $F_1$ and $F_2$ can be extended trivially to the spacetime region $M$ and there be multiplied to a classical observable $F_1\cdot F_2:K_{M}\to\R$. The composition property of the Feynman path integral now implies the identity
\begin{equation}
 \rho_{M}^{F_1\cdot F_2}=\rho_{M_1}^{F_1}\aglue \rho_{M_2}^{F_2} .
\label{eq:corprod}
\end{equation}
That is, there is a direct correspondence between the product of classical observables and the spacetime composition of quantum observables. This \emph{composition correspondence}, as we shall call it, is not to be confused with what is usually meant with the term ``correspondence principle'' such as a relation between the commutator of operators and the Poisson bracket of classical observables that these are representing. Indeed, at a careless glance these concepts might even seem to be in contradiction.

Consider now in Minkowski space a region $M=[t_1,t_2]\times \R^3$, where $t_1<t_2$. Then, $\cH_{\partial M}=\cH_{\Sigma_1}\ctens\cH_{\Sigma_2}^*$ with notation as in Section~\ref{sec:stdampl}. Consider a classical observable $F_{x_1,\dots,x_n}:K_M\to \R$ that encodes an $n$-point function,\footnote{For simplicity we use notation here that suggests a real scalar field.}
\begin{equation}
 F_{x_1,\dots,x_n}:\phi\mapsto \phi(x_1)\cdots\phi(x_n),
\end{equation}
where $x_1,\dots,x_n\in M$. Given an initial state $\psi\in\cH_{\Sigma_1}$ at time $t_1$ and a final state $\eta\in\cH_{\Sigma_2}$ at time $t_2$, the quantization of $F_{x_1,\dots,x_n}$ according to formula (\ref{eq:opi}) can be written in the more familiar form
\begin{multline}
 \rho_M^{F_{x_1,\dots,x_n}}(\psi\tens\iota(\eta))\\
=\int_{K_{M}} \psi(\phi|_{\Sigma_1})\overline{\eta(\phi|_{\Sigma_2})}\, \phi(x_1)\cdots\phi(x_n)\, e^{\im S_M(\phi)}\,\xd\mu_{M}(\phi)\\
 =\langle \eta, \tord \tilde{\phi}(x_1)\cdots\tilde{\phi}(x_n) e^{-\im \int_{t_1}^{t_2} \tilde{H}(t)\,\xd t}\psi\rangle,
\label{eq:qtord}
\end{multline}
where $\tilde{\phi}(x_i)$ are the usual quantizations of the classical observables $\phi\mapsto \phi(x_i)$, $\tilde{H}(t)$ is the usual quantization of the Hamiltonian operator at time $t$ and $\tord$ signifies time-ordering. Thus, in familiar situations the prescription (\ref{eq:opi}) really is the ``usual'' quantization performed in quantum field theory, but with time-ordering of operators. From formula (\ref{eq:qtord}) the correspondence property (\ref{eq:corprod}) is also clear, although in the more limited context of temporal composition. We realize thus the goal, mentioned in the introduction, of implementing the time-ordered product as more fundamental than the non-commutative operator product.

For a linear field theory, it turns out that the quantization prescription encoded in (\ref{eq:opi}) exhibits an interesting factorization property with respect to coherent states. We consider the simple setting of a massive free scalar field theory in Minkowski space with equal-time hypersurfaces. Recall (\cite{CoOe:smatrixgbf} equation (26)) that a coherent state in the Schrödinger representation at time $t$ can be written as
\begin{equation}
\psi_{t,\eta} (\varphi) \defeq C_{t,\eta} \exp \left( \int
\frac{\xd^3 x\,\xd^3 k}{(2 \pi)^3} \, \eta(k) \, e^{-\im(E t-k x)} \,
\varphi(x)\right)\, \psi_0(\varphi) ,
\end{equation}
where $\eta$ is a complex function on momentum space encoding a solution of the Klein-Gordon equation. $\psi_0$ is the vacuum wave function and $C_{t,\eta}$ is a normalization constant. Consider as above an initial time $t_1$, a final time $t_2>t_1$ and the region $M\defeq [t_1,t_2]\times\R^3$ in Minkowski space. Let $F:K_M\to\C$ represent a classical observable. Evaluating the quantized observable map $\rho_M^F$ on an initial coherent state encoded by $\eta_1$ and a final coherent state encoded by $\eta_2$ yields,
\begin{align}
 & \rho_M^F\left(\psi_{t_1,\eta_1}\tens\overline{\psi_{t_2,\eta_2}}\right)\nonumber\\
 & = C_{t_1,\eta_1}\overline{C_{t_2,\eta_2}} \int_{K_M}
 \psi_0(\varphi_1)\,\overline{\psi_0(\varphi_2)}\nonumber\\
 &\quad \exp \left( \int \frac{\xd^3 x\,\xd^3 k}{(2 \pi)^3} \,
\left(\eta_1(k) \,
e^{-\im(E t_1-k x)} \, \varphi_1(x)
+\overline{\eta_2(k)} \,
e^{\im(E t_2-k x)} \, \varphi_2(x)\right)\right)\nonumber\\
&\quad F(\phi)\,e^{\im S_M(\phi)}\,\xd\mu_M(\phi)\nonumber\\
& = \rho_M\left(\psi_{t_1,\eta_1}\tens\overline{\psi_{t_2,\eta_2}}\right)
 \int_{K_M}
 \psi_0(\varphi_1)\,\overline{\psi_0(\varphi_2)}\,
F(\phi+\hat{\eta})\,e^{\im S_M(\phi)}\,\xd\mu_M(\phi) .
\label{eq:sfaci}
\end{align}
Here, $\varphi_i$ denote the restrictions of the configuration $\phi$ to time $t_i$. To obtain the second equality we have shifted the integration variable $\phi$ by
\begin{equation}
\hat{\eta}(t,x)\defeq
\int \frac{\xd^3 k}{(2\pi)^3 2 E}
\left(\eta_1(k) e^{-\im (E t- k x)}
+\overline{\eta_2(k)} e^{\im(E t- k x)}\right)
\label{eq:classcoh}
\end{equation}
and used the conventions of \cite{CoOe:smatrixgbf}. Note that $\hat{\eta}$ is a \emph{complexified} classical solution in $M$ determined by $\eta_1$ and $\eta_2$. We have supposed that $F$ naturally extends to a function on the \emph{complexified} configuration space $K_M^{\C}$. Viewing the function $\phi\to F(\phi+\hat{\eta})$ as a new observable $F^{\hat{\eta}}$, the remaining integral in (\ref{eq:sfaci}) can be interpreted in terms of (\ref{eq:opi}) and we obtain the factorization identity
\begin{equation}
\rho_M^F\left(\psi_{t_1,\eta_1}\tens\overline{\psi_{t_2,\eta_2}}\right)= \rho_M\left(\psi_{t_1,\eta_1}\tens\overline{\psi_{t_2,\eta_2}}\right)
\rho_M^{F^{\hat{\eta}}}(\psi_0\tens\psi_0) .
\label{eq:sfactid}
\end{equation}
That is, the quantum observable map evaluated on a pair of coherent states factorizes into the plain amplitude for the same pair of states and the quantum observable map for a shifted observable evaluated on the vacuum. Note that the second term on the right hand side here is a vacuum expectation value.

It turns out that factorization identities analogous to (\ref{eq:sfactid}) are generic rather than special to the types of hypersurfaces and regions considered here. We will come back to this issue in the next section, where also the role of the complex classical solution $\hat{\eta}$ will become clearer from the point of view of holomorphic quantization. For the moment let us consider the particularly simple case where $F$ is a linear observable. In this case $F^{\hat{\eta}}(\phi)=F(\phi)+F(\hat{\eta})$ and the second term on the right hand side of (\ref{eq:sfactid}) decomposes into a sum of two terms,
\begin{equation}
 \rho_M^{F^{\hat{\eta}}}(\psi_0\tens\psi_0)=\rho_M^F(\psi_0\tens\psi_0) + F(\hat{\eta}) \rho_M(\psi_0\tens\psi_0).
\end{equation}
The first term on the right hand side is a one-point function which vanishes in the present case of a linear field theory. ($F$ is antisymmetric under exchange of $\phi$ and $-\phi$, while the other expressions in (\ref{eq:opi}) are symmetric.) The second factor in the second term is the amplitude of the vacuum and hence equal to unity. Thus, in the case of a linear observable (\ref{eq:sfactid}) simplifies to
\begin{equation}
\rho_M^F\left(\psi_{t_1,\eta_1}\tens\overline{\psi_{t_2,\eta_2}}\right)= F(\hat{\eta}) \rho_M\left(\psi_{t_1,\eta_1}\tens\overline{\psi_{t_2,\eta_2}}\right).
\label{eq:sflin}
\end{equation}

\subsection{Holomorphic quantization}

A more rigorous quantization scheme that produces a GBF from a classical field theory is the holomorphic quantization scheme introduced in \cite{Oe:holomorphic}. It is based on ideas from geometric quantization and its Hilbert spaces are versions of ``Fock representations''. An advantage of this scheme is that taking an axiomatically described classical field theory as input, it produces a GBF as output that can be rigorously proved to satisfy the core axioms of Section~\ref{sec:caxioms}. A shortcoming so far is that only the case of linear field theory has been worked out.

Concretely, the classical field theory is to be provided in the form of a real vector space $L_{\Sigma}$ of (germs of) solutions near each hypersurface $\Sigma$. Moreover, for each region $M$ there is to be given a subspace $L_{\tilde M}$ of the space $L_{\partial M}$ of solutions on the boundary of $M$. This space $L_{\tilde M}$ has the interpretation of being the space of solutions in the interior of $M$ (restricted to the boundary). Also, the spaces $L_{\Sigma}$ carry non-degenerate symplectic structures $\omega_\Sigma$ as well as complex structures $J_\Sigma$. Moreover, for each hypersurface $\Sigma$, the symplectic and complex structures combine to a complete real inner product $g_\Sigma(\cdot,\cdot)=2\omega_\Sigma(\cdot,J_\Sigma\cdot)$ and to a complete complex inner product $\{\cdot,\cdot\}_\Sigma=g_\Sigma(\cdot,\cdot)+2\im\omega_\Sigma(\cdot,\cdot)$. Another important condition is that the subspace $L_{\tilde{M}}\subseteq L_{\partial M}$ is Lagrangian with respect to the symplectic structure $\omega_{\partial M}$.

The Hilbert space $\cH_\Sigma$ associated with a hypersurface $\Sigma$ is the space of \emph{holomorphic} square-integrable functions on $\hat{L}_\Sigma$ with respect to a Gaussian measure $\nu_\Sigma$.\footnote{The space $\hat{L}_\Sigma$ is a certain extension of the space $L_\Sigma$, namely the algebraic dual of its topological dual. Nevertheless, due to Theorem~3.18 of \cite{Oe:holomorphic} it is justified to think of wave functions $\psi$ as functions merely on $L_\Sigma$ rather than on $\hat{L}_\Sigma$, and to essentially ignore the distinction between $L_\Sigma$ and $\hat{L}_\Sigma$.} That is, the inner product in $\cH_\Sigma$ is
\begin{equation}
\langle \psi',\psi\rangle_\Sigma\defeq\int_{\hat{L}_\Sigma} \psi(\phi)\overline{\psi'(\phi)}\,\xd\nu_\Sigma(\phi) .
\label{eq:hsip}
\end{equation} 
Heuristically, the measure $\nu_\Sigma$ can be understood as
\begin{equation}
\xd\nu_\Sigma(\phi)\approx \exp\left(-\frac{1}{2}g_{\partial M}(\phi,\phi)\right)\xd\mu_\Sigma(\phi) ,
\end{equation}
where $\mu_\Sigma$ is a fictitious translation invariant measure on ${\hat{L}_\Sigma}$. The space $\cH_\Sigma$ is essentially the Fock space constructed from $L_\Sigma$ viewed as a 1-particle space with the inner product $\{\cdot,\cdot\}_\Sigma$.

The amplitude map $\rho_M:\cH_{\partial M}\to\C$ associated with a region $M$ is given by the integral formula
\begin{equation}
\rho_M(\psi)\defeq\int_{\hat{L}_{\tilde{M}}} \psi(\phi)\,\xd\nu_{\tilde{M}}(\phi).
\end{equation}
The integration here is over the space $\hat{L}_{\tilde M}\subseteq \hat{L}_{\partial M}$ of solutions in $M$ with the measure $\nu_{\tilde{M}}$, which heuristically can be understood as
\begin{equation}
\xd\nu_{\tilde{M}}(\phi)\approx \exp\left(-\frac{1}{4}g_{\partial M}(\phi,\phi)\right)\xd\mu_{\tilde{M}}(\phi) ,
\end{equation}
where again $\mu_{\tilde{M}}$ is a fictitious translation invariant measure on $\hat{L}_{\tilde{M}}$.

Particularly useful in the holomorphic quantization scheme turn out to be the coherent states that are associated to classical solutions near the corresponding hypersurface. On a hypersurface $\Sigma$ the coherent state $K_\xi\in\cH_\Sigma$ associated to $\xi\in L_\Sigma$ is given by the wave function
\begin{equation}
 K_\xi(\phi)\defeq\exp\left(\frac{1}{2}\{\xi,\phi\}_\Sigma\right) .
\end{equation}
The natural vacuum, which we denote by $\one$, is the constant wave function of unit value. Note that $\one=K_0$.

\subsubsection{Creation and annihilation operators}
\label{sec:crean}

One-particle states on a hypersurface $\Sigma$ are represented by non-zero continuous complex-linear maps $p:L_\Sigma\to\C$, where complex-linearity here implies $p(J\xi)=\im p(\xi)$. By the Riesz Representation Theorem such maps are thus in one-to-one correspondence with non-zero elements of $L_\Sigma$. Concretely, for a non-zero element $\xi\in L_\Sigma$ the corresponding one-particle state is represented by the wave function $p_\xi\in\cH_\Sigma$ given by
\begin{equation}
p_\xi(\phi)=\frac{1}{\sqrt{2}}\{\xi,\phi\}_\Sigma .
\end{equation}
The normalization is chosen here such that $\|p_\xi\|=\|\xi\|$. Physically distinct one-particle states thus correspond to the distinct rays in $L_\Sigma$, viewed as a complex Hilbert space.
An $n$-particle state is represented by a (possibly infinite) linear combination of the product of $n$ wave functions of this type. The creation operator $a^\dagger_\xi$ for a particle state corresponding to $\xi\in L_\Sigma$ is given by multiplication,
\begin{equation}
(a^\dagger_\xi\psi)(\phi) = p_\xi(\phi)\psi(\phi)= \frac{1}{\sqrt{2}}\{\xi,\phi\}_\Sigma\psi(\phi) .
\label{eq:coact}
\end{equation}
The corresponding annihilation operator is the adjoint. Using the reproducing property of the coherent states $K_\phi\in\cH_\Sigma$ we can write it as,
\begin{equation}
(a_\xi\psi)(\phi) = \langle K_\phi, a_\xi \psi \rangle_\Sigma = \langle a^\dagger_\xi K_\phi, \psi \rangle_\Sigma .
\end{equation}
Note in particular, that the action of an annihilation operator on a coherent state is by multiplication,
\begin{equation}
a_\xi K_\phi = \frac{1}{\sqrt{2}}\{\phi,\xi\}_\Sigma K_\phi .
\label{eq:aocact}
\end{equation}
For $\xi,\eta\in L_\Sigma$ the commutation relations are, as usual,
\begin{equation}
[a_\xi,a_\eta^\dagger]=\{\eta,\xi\}_\Sigma,\qquad [a_\xi,a_\eta]=0,\qquad [a_\xi^\dagger,a_\eta^\dagger]=0 .
\end{equation}

\subsubsection{Berezin-Toeplitz quantization}

A natural way to include observables into this quantization scheme seems to be the following. We model a classical observable $F$ on a spacetime region $M$ as a map $L_{\tilde{M}}\to\C$ (or $L_{\tilde{M}}\to\R$) and define the associated quantized observable map via
\begin{equation}
\rho_M^{\ano{F}}(\psi)\defeq \int_{\hat{L}_{\tilde{M}}} \psi(\phi) F(\phi)\,\xd\nu_{\tilde{M}}(\phi) .
\label{eq:bt}
\end{equation}
To bring this into a more familiar form, we consider, as in Section~\ref{sec:recobsexp}, the special case of an empty region $\hat{\Sigma}$, given geometrically by a hypersurface $\Sigma$. Then, for $\psi_1,\psi_1\in\cH_{\Sigma}$ encoding ``initial'' and ``final'' state we have
\begin{equation}
\rho_{\hat{\Sigma}}^{\ano{F}}(\psi_1\tens\iota(\psi_2))=\int_{\hat{L}_{\Sigma}} \psi_1(\phi) \overline{\psi_2(\phi)} F(\phi)\,\xd\nu_{\Sigma}(\phi) .
\label{eq:bths}
\end{equation}
We can interpret this formula as follows: The wave function $\psi_1$ is multiplied by the function $F$. The resulting function is an element of the Hilbert space $\rL^2(\hat{L}_\Sigma,\nu_\Sigma)$ (supposing $F$ to be essentially bounded), but not of the subspace $\cH_\Sigma$ of holomorphic functions. We thus project back onto this subspace and finally take the inner product with the state $\psi_2$. This is precisely accomplished by the integral. We may recognize this as a version of Berezin-Toeplitz quantization, where in the language of Berezin \cite{Ber:covcont} the function $F$ is the contravariant symbol of the operator $\tilde{F}$ that is related to $\rho_{\hat{\Sigma}}^{\ano{F}}$ by formula (\ref{eq:obscor}). That is,
\begin{equation}
\rho_{\hat{\Sigma}}^{\ano{F}}(\psi_1\tens\iota(\psi_2))=\langle\psi_2,\tilde{F}\psi_1\rangle_\Sigma .
\label{eq:opcor}
\end{equation}
In the following we shall refer to the prescription encoded in (\ref{eq:bt}) simply as Berezin-Toeplitz quantization.

Note that any complex valued continuous real-linear observable $F:L_\Sigma\to\C$ can be decomposed into its holomorphic (complex linear) and anti-holomorphic (complex conjugate linear) part
\begin{equation}
 F(\phi)=F^+(\phi)+ F^-(\phi),\quad\text{where}\quad F^{\pm}(\phi)=\frac{1}{2}\left(F(\phi)\mp \im F(J_\Sigma\phi)\right) .
\label{eq:declinobs}
\end{equation}
If we consider real valued observables only, we can parametrize them by elements of $L_\Sigma$ due to the Riesz Representation Theorem. (In the complex valued case the parametrization is by elements of $L_\Sigma^{\C}$, the complexification of $L_\Sigma$, instead.) If we associate to $\xi\in L_\Sigma$ the real linear observable $F_\xi$ given by
\begin{gather}
 F_\xi(\phi)\defeq \sqrt{2}\, g_\Sigma(\xi,\phi),\quad\text{then} \label{eq:paramlinobs}\\
 F^+_\xi(\phi)=\frac{1}{\sqrt{2}}\{\xi,\phi\}_\Sigma,\quad
 F^-_\xi(\phi)=\frac{1}{\sqrt{2}}\{\phi,\xi\}_\Sigma.
\end{gather}
Using the results of Section~\ref{sec:crean}, we see that $F^\pm_\xi$ quantized according to the prescription (\ref{eq:bths}) and expressed in terms of operators $\tilde{F}^\pm_\xi$ as in (\ref{eq:opcor}) yields
\begin{equation}
 \tilde{F}^+_\xi=a^\dagger_\xi,\quad \tilde{F}^-_\xi=a_\xi,\quad\text{and for the sum,}
\quad \tilde{F}_\xi=a^\dagger_\xi+a_\xi .
\end{equation}

Consider now $n$ real-linear observables $F_1,\dots,F_n:L_\Sigma\to\C$. We shall be interested in the \emph{antinormal ordered product} of the corresponding operators $\tilde{F}_1,\dots,\tilde{F}_n$, which we denote by $\ano{\tilde{F}_1\cdots\tilde{F}_n}$. To evaluate matrix elements of this antinormal ordered product we decompose the observables $F_i$ according to (\ref{eq:declinobs}) into holomorphic and anti-holomorphic parts, corresponding to creation operators and annihilation operators respectively. The creation operators $\tilde{F}_i^+$ then act on wave functions by multiplication with $F_i^+$ according to (\ref{eq:coact}). Converting the annihilation operators into creation operators by moving them to the left-hand side of the inner product, we see that these correspondingly contribute factors $F_i^-$ in the inner product (\ref{eq:hsip}). We obtain,
\begin{align*}
\langle \psi_2, \ano{\tilde{F}_1\cdots\tilde{F}_n} \psi_1\rangle_\Sigma &
=\langle \psi_2, \ano{\prod_{i=1}^n(\tilde{F}^+_i+\tilde{F}_i^-)} \psi_1\rangle_\Sigma\\
& =\int_{\hat{L}_\Sigma} \overline{\psi_2(\phi)} \left(\prod_{i=1}^n(F^+_i(\phi)+F_i^-(\phi))\right)
  \psi_1(\phi)\,\xd\nu_\Sigma(\phi) \\
& =\int_{\hat{L}_\Sigma} \overline{\psi_2(\phi)} F_1(\phi)\cdots F_n(\phi)
  \psi_1(\phi)\,\xd\nu_\Sigma(\phi) .
\end{align*}
Setting $F\defeq F_1\cdots F_n$ this coincides precisely with the right-hand side of (\ref{eq:bths}). Thus, in the case of a hypersurface (empty region) the Berezin-Toeplitz quantization precisely implements antinormal ordering.

Remarkably, the Berezin-Toeplitz quantization shares with the Schrö\-din\-ger-Feynman quantization the factorization property exhibited in equation (\ref{eq:sfactid}). In fact, it is in the present context of holomorphic quantization that this property attains a strikingly simple form. In order to state it rigorously, we need a bit of technical language.
For a map $F:L_{\tilde{M}}\to\C$ and an element $\xi\in L_{\tilde{M}}$ we denote by $F^\xi:L_{\tilde{M}}\to\C$ the translated map $\phi\mapsto F(\phi+\xi)$. We say that $F:L_{\tilde{M}}\to\C$ is \emph{analytic} iff for each pair $\phi,\xi\in L_{\tilde{M}}$ the map $z\mapsto F(\phi+z\xi)$ is real analytic. We denote the induced extension $L_{\tilde{M}}^\C\to\C$ also by $F$, where $L_{\tilde{M}}^\C$ is the complexification of $L_{\tilde{M}}$. We say that $F:L_{\tilde{M}}\to\C$ is analytic and \emph{sufficiently integrable} iff for any $\eta\in L_{\tilde{M}}^\C$ the map $F^\eta$ is integrable in $(\hat{L}_{\tilde{M}},\nu_{\tilde{M}})$. We recall (Lemma~4.1 of \cite{Oe:holomorphic}) that elements $\xi$ of $L_{\partial M}$ decompose uniquely as $\xi=\xi^{\textrm{R}}+J_\Sigma \xi^{\textrm{I}}$, where $\xi^{\textrm{R}}, \xi^{\textrm{I}}$ are elements of $L_{\tilde{M}}$.

\begin{prop}[Coherent Factorization Property]
Let $F:L_M\to\C$ be analytic and sufficiently integrable. Then, for any $\xi\in L_{\partial M}$ we have
\begin{equation}
 \rho_M^{\ano{F}}(K_\xi)=\rho_M(K_\xi)\, \rho_M^{\ano{F^{\hat{\xi}}}}(\one),
\label{eq:cohfact}
\end{equation}
where $\hat{\xi}\in L_{\tilde{M}}^{\C}$ is given by $\hat{\xi}=\xi^{\mathrm{R}}-\im\xi^{\mathrm{I}}$.
\end{prop}
\begin{proof}
Recall that for $\phi\in\hat{L}_{\tilde{M}}$ we can rewrite the wave function of the coherent state $K_\xi$ as follows,
\begin{equation}
 K_\xi(\phi)=\exp\left(\frac{1}{2}g_{\partial M}(\xi^{\textrm{R}}, \phi)-\frac{\im}{2}g_{\partial M}(\xi^{\textrm{I}}, \phi)\right) .
\label{eq:cohdec}
\end{equation}
We restrict first to the special case $\xi\in L_{\tilde{M}}$, i.e., $\xi^{\textrm{I}}=0$. Translating the integrand by $\xi$ (using Proposition~3.11 of \cite{Oe:holomorphic}) we find
\begin{align*}
& \int_{\hat{L}_{\tilde{M}}} F(\phi) \exp\left(\frac{1}{2} g_{\partial M}(\xi,\phi)\right)\xd\nu(\phi)  \\
& = \int_{\hat{L}_{\tilde{M}}} F(\phi+\xi) \exp\left(\frac{1}{2} g_{\partial M}(\xi,\phi+\xi)-\frac{1}{4}g_{\partial M}(2\phi+\xi,\xi)\right)\xd\nu(\phi) \\
& =  \exp\left(\frac{1}{4}g_{\partial M}(\xi,\xi)\right) \int_{\hat{L}_{\tilde{M}}} F(\phi+\xi)\,\xd\nu(\phi)\\
& = \rho_M(K_\xi)\,\rho_M^{\ano{F^\xi}}(\one)
\end{align*}
In order to work out the general case, we follow the strategy outlined in the proof of Proposition~4.2 of \cite{Oe:holomorphic}: We replace the $\im$ in (\ref{eq:cohdec}) by a complex parameter and note that all relevant expressions are holomorphic in this parameter. This must also hold for the result of the integration performed above. But performing the integration is straightforward when this parameter is real, since we can then combine both terms in the exponential in (\ref{eq:cohdec}). On the other hand, a holomorphic function is completely determined by its values on the real line. This leads to the stated result.
\end{proof}

It is clear at this point that equation (\ref{eq:sfactid}) is just a special case of (the analogue for $\rho_M^F$ of) equation (\ref{eq:cohfact}). Indeed, it turns out that with a suitable choice of complex structure (see \cite{Oe:holomorphic}) the complexified classical solution $\hat{\eta}$ given by (\ref{eq:classcoh}) decomposes precisely as $\hat{\eta}=\eta^{\mathrm{R}}-\im\eta^{\mathrm{I}}$.\footnote{We differ here slightly from the conventions in \cite{Oe:holomorphic} to obtain exact agreement.} From here onwards we shall say that a quantization scheme satisfying equation (\ref{eq:cohfact}) has the \emph{coherent factorization property}.

The coherent factorization property may also be interpreted as suggesting an intrinsic definition of a \emph{real} observable in the quantum setting. It is clear that quantum observable maps must take values in the complex numbers and not merely in the real numbers since for example the amplitude map is a special kind of quantum observable map.\footnote{Proposition~4.2 of \cite{Oe:holomorphic} implies that amplitude maps generically take complex and not merely real values.} On the other hand, we have in axiom (O1) deliberately only required that the observable maps in a region $M$ form a real vector space $\obs_M$, to allow for a restriction to ``real'' observables, analogous to hermitian operators in the standard formulation and to real valued maps in the classical theory. Of course, given a quantization prescription such as (\ref{eq:opi}) or (\ref{eq:bt}), we can simply restrict the quantization to real classical observables. However, equation (\ref{eq:cohfact}) suggests a more intrinsic definition in case of availability of coherent states. Namely, we could say that a quantum observable map is \emph{real} iff its evaluation on a coherent state $K_\xi$ associated to any element $\xi$ in the subspace $L_{\tilde M}\subseteq L_{\partial M}$ yields a real multiple of the amplitude map evaluated on the same coherent state. Note that this characterization is closed under real linear combinations. Also, if a quantization scheme satisfies the coherent factorization property, this characterization coincides with the condition for the classical observable to be real valued, as is easily deduced using the completeness of the coherent states.

Let us briefly return to the special case of linear observables. Suppose that $F:L_{\tilde M}\to\R$ is linear (implying analytic) and sufficiently integrable. We evaluate the Berezin-Toeplitz quantum observable map $\rho_M^{\ano{F}}$ on the coherent state $K_\xi$ associated to $\xi\in L_{\partial M}$. As above we define $\hat{\xi}\in L_{\tilde{M}}^{\C}$ as $\hat{\xi}=\xi^{\mathrm{R}}-\im\xi^{\mathrm{I}}$. Using the coherent factorization property (\ref{eq:cohfact}) as well as linearity of $F$ we obtain
\begin{equation}
\rho_M^{\ano{F}}(K_\xi)=\rho_M(K_\xi) \rho_M^{\ano{F^{\hat{\xi}}}}(\one)=\rho_M(K_\xi)
 \left(\rho_M^{\ano{F}}(\one)+F(\hat{\xi})\rho_M(\one) \right) .
\end{equation}
The first term in brackets vanishes by inspection of (\ref{eq:bt}) due to anti-symmetry of $F$ under exchange of $\phi$ and $-\phi$, while $\rho_M(\one)=1$. Thus, analogous to (\ref{eq:sflin}) we obtain
\begin{equation}
\rho_M^{\ano{F}}(K_\xi)=F(\hat{\xi})\rho_M(K_\xi) .
\label{eq:btlin}
\end{equation}

Supposing that the amplitudes of coherent states coincide between the Schrödinger-Feynman scheme and the holomorphic scheme (that is, if the complex structure of the holomorphic scheme and the vacuum of the Schrö\-din\-ger-Feynman scheme are mutually adapted), also the quantization of linear observables according to (\ref{eq:opi}) coincides with that of (\ref{eq:bt}). Nevertheless, the quantization of non-linear observables is necessarily different. For one, classical observables in the Schrödinger-Feynman scheme are defined on configuration spaces rather than on spaces of solutions. Indeed, the quantization of observables that coincide when viewed merely as functions on solutions differs in general. However, it is also clear that the Berezin-Toeplitz quantization cannot satisfy the composition correspondence property (\ref{eq:corprod}) that is satisfied by the Schrödinger-Feynman scheme. Indeed, consider adjacent regions $M_1$ and $M_2$ that can be glued to a joint region $M$. Then, the classical observables in the disjoint region induce classical observables in the glued region, but not the other way round. While the former are functions on $L_{\tilde{M}_1}\times L_{\tilde{M}_2}$ the latter are functions on the subspace $L_{\tilde{M}}\subseteq L_{\tilde{M}_1}\times L_{\tilde{M}_2}$. In spite of the summation involved in axiom (O2b), one can use this to cook up a contradiction to the composition correspondence property (\ref{eq:corprod}).
It is easy to see how this problem is avoided in the Schrödinger-Feynman scheme: There, classical observables in a region $M$ are functions on the space of field configurations $K_M$, which is much larger than the space of classical solutions $L_M$ and permits the ``decoupling'' of observables in adjacent regions. Indeed, the present considerations indicate that in order for a quantization scheme to satisfy the composition correspondence property, this kind of modification of the definition of what constitutes a classical observable is a necessity.

The composition correspondence property suggests also a different route to quantization of observables: We may consider a quantization scheme only for linear observables and impose the composition correspondence property to \emph{define} a quantization scheme for more general observables. In the Berezin-Toeplitz case this would lead to a scheme equivalent to the path integral (\ref{eq:opi}). However, recall that the composition of quantum observable maps is only possible between disjoint regions. On the other hand, well known difficulties (related to renormalization) arise also for the path integral (\ref{eq:opi}) when considering field observables at coincident points.

\subsubsection{Normal ordered quantization}

Consider a hypersurface $\Sigma$ and linear observables $F_1,\dots,F_n:L_\Sigma\to\C$ in the associated empty region $\hat{\Sigma}$. Consider now the \emph{normal ordered product} $\no{\tilde{F}_1\cdots \tilde{F}_n}$ and its matrix elements. These matrix elements turn out to be particularly simple for coherent states. To evaluate these we decompose the maps $F_1,\dots,F_n$ into holomorphic (creation) and anti-holomorphic (annihilation) parts according to (\ref{eq:declinobs}). The annihilation operators act on coherent states simply by multiplication according to (\ref{eq:aocact}). The creation operators on the other hand can be converted to annihilation operators by moving them to the left-hand side of the inner product. We find,
\begin{equation}
 \langle K_\eta, \no{F_1\cdots F_n} K_\xi\rangle_\Sigma
 =\prod_{i=1}^n \left(F_i^+(\eta)+F_i^-(\xi)\right) \langle K_\eta, K_\xi\rangle_\Sigma .
\label{eq:nprodcohid}
\end{equation}
While this expression looks quite simple, it can be further simplified by taking serious the point of view that $\hat{\Sigma}$ is an (empty) region. Hence, $K_\xi\tens\iota(K_\eta)$ is really the coherent state $K_{(\xi,\eta)}\in\cH_{\partial \hat{\Sigma}}$ associated to the solution $(\xi,\eta)\in L_{\partial \hat{\Sigma}}$. As above we may decompose $(\xi,\eta)=(\xi,\eta)^{\textrm{R}}+J_{\partial\hat{\Sigma}}(\xi,\eta)^{\textrm{I}}$, where $(\xi,\eta)^{\textrm{R}},(\xi,\eta)^{\textrm{I}}\in L_{\tilde{\hat{\Sigma}}}$. Identifying $L_{\tilde{\hat{\Sigma}}}$ with $L_\Sigma$ (and taking into account $J_{\partial\hat{\Sigma}}=(J_\Sigma,-J_{\Sigma})$) we have
\begin{equation}
 (\xi,\eta)^{\textrm{R}}=\frac{1}{2}(\xi+\eta),\quad (\xi,\eta)^{\textrm{I}}=-\frac{1}{2}(J_\Sigma\xi-J_\Sigma\eta) .
\end{equation}
But observe,
\begin{multline}
 F_i^+(\eta)+F_i^-(\xi)=\frac{1}{2}\left(F_i(\eta+\xi)-\im F_i(J_\Sigma(\eta-\xi))\right)\\
 = F_i\left((\xi,\eta)^{\textrm{R}}\right)-\im F_i\left((\xi,\eta)^{\textrm{I}}\right)
 = F_i\left((\xi,\eta)^{\textrm{R}}-\im (\xi,\eta)^{\textrm{I}}\right),
\end{multline}
where in the last step we have extended the domain of $F_i$ from $L_{\tilde{\hat{\Sigma}}}$ to its complexification $L_{\tilde{\hat{\Sigma}}}^{\C}$.

Defining a quantum observable map encoding the normal ordered product
\begin{equation}
 \rho_{\hat{\Sigma}}^{\no{F_1\cdots F_n}}(\psi_1\tens\iota(\psi_2))\defeq \langle \psi_2, \no{\tilde{F}_1\cdots \tilde{F}_n} \psi_1\rangle_\Sigma ,
\label{eq:nquantops}
\end{equation}
the identity (\ref{eq:nprodcohid}) becomes thus
\begin{equation}
 \rho_{\hat{\Sigma}}^{\no{F_1\cdots F_n}}(K_{(\xi,\eta)})
 =\prod_{i=1}^n F_i\left((\xi,\eta)^{\textrm{R}}-\im (\xi,\eta)^{\textrm{I}}\right) \rho_{\hat{\Sigma}}(K_{(\xi,\eta)}) .
\label{eq:nprodcoh}
\end{equation}
Note that in the above expression the fact that we consider an empty region rather than a generic region is no longer essential. Rather, we may consider a region $M$ and replace $(\xi,\eta)$ by some solution $\phi\in L_{\tilde{M}}$. Also there is no longer a necessity to write the observable explicitly as a product of linear observables. A generic observable $F:L_{\tilde{M}}\to\C$ that has the analyticity property will do. We obtain,
\begin{equation}
 \rho_{M}^{\no{F}}(K_{\phi})
 \defeq F(\hat{\phi}) \rho_{M}(K_{\phi}),
\label{eq:nquant}
\end{equation}
where $\hat{\phi}\defeq\phi^{\textrm{R}}-\im\phi^{\textrm{I}}$. We may take this now as the \emph{definition} of a quantization prescription that we shall call \emph{normal ordered quantization}. It coincides with, and provides an extremely concise expression for the usual concept of normal ordering in the case when $M$ is the empty region associated to a hypersurface.

Interestingly, expression (\ref{eq:nquant}) also coincides with expression (\ref{eq:btlin}) and with expression (\ref{eq:sflin}). However, the latter two expressions were only valid in the case where $F$ is linear. So, unsurprisingly, we obtain agreement of normal ordered quantization with Berezin-Toeplitz quantization and with Schrödinger-Feynman quantization in the case of linear observables, while they differ in general. Remarkably, however, normal ordered quantization shares with these other quantization prescriptions the coherent factorization property (\ref{eq:cohfact}). To see this, note using (\ref{eq:nquant}),
\begin{equation}
 \rho_M^{\no{F^{\hat{\phi}}}}(\one)=F^{\hat{\phi}}(0)\rho_M(\one)=F(\hat{\phi}) .
\end{equation}

\subsubsection{Geometric quantization}

Since the holomorphic quantization scheme draws on certain ingredients of geometric quantization, it is natural to also consider what geometric quantization has to say about the quantization of observables \cite{Woo:geomquant}. For hypersurfaces $\Sigma$ (empty regions $\hat{\Sigma}$) the geometric quantization of a classical observable $F:L_\Sigma\to\C$ is given by an operator $\check{F}:\cH_\Sigma\to\cH_\Sigma$. If the observable $F$ preserves the polarization (which is the case for example for linear observables), then $\check{F}$ is given by the formula
\begin{equation}
 \check{F} \psi=-\im\, \xd\psi(\mathbf{F})-\theta(\mathbf{F}) \psi+ F \psi .
\label{eq:geomquantobs}
\end{equation}
Here $\mathbf{F}$ denotes the Hamiltonian vector field generated by the function $F$, $\theta$ is the symplectic potential given here by $\theta_\eta(\Phi)=-\frac{\im}{2}\{\eta,\Phi_\eta\}_\Sigma$, and $\xd\psi$ is the exterior derivative of $\psi$.

Consider a real-linear observable $F:L_\Sigma\to\C$. Without explaining the details we remark that for the holomorphic part $F^+$ (recall (\ref{eq:declinobs})) we obtain $\xd\psi(\mathbf{F^+})=0$ as well as $\theta(\mathbf{F^+})=0$. On the other hand, for the anti-holomorphic part $F^-$ we have $\theta(\mathbf{F^-})=F^-$. This simplifies (\ref{eq:geomquantobs}) to
\begin{equation}
 \check{F} \psi=-\im\, \xd\psi(\mathbf{F^-}) + F^+ \psi .
\end{equation}
Setting $F$ equal to $F_\xi$ given by (\ref{eq:paramlinobs}) for $\xi\in L_\Sigma$ results in $F^+ \psi = a^\dagger_\xi \psi$ and $-\im\, \xd\psi(\mathbf{F^-})=a_\xi \psi$. That, is the operator $\check{F}$ coincides with the operator $\tilde{F}$ obtained by quantizing $F$ with any of the other prescriptions discussed. On the other hand, the quantization of non-linear observables will in general differ from any of the other prescriptions. Indeed, it is at present not even clear whether or how the prescription (\ref{eq:geomquantobs}) can be generalized to non-empty regions.

\subsection*{Acknowledgments}
I would like to thank the organizers of the \emph{Regensburg Conference 2010: Quantum Field Theory and Gravity} for inviting me to present aspects of this work as well as for providing a stimulating and well-organized conference. This work was supported in part by CONACyT grant 49093.

\bibliographystyle{amsordx} 
\bibliography{stdrefs}
\end{document}